\documentclass[12pt,a4paper]{article}
\usepackage[utf8]{inputenc}
\usepackage{amsmath,amsfonts,amssymb,amsthm}
\usepackage{graphicx}
\usepackage[export]{adjustbox}
\usepackage{subcaption}
\usepackage{verbatim}
\usepackage{authblk}
\usepackage{tikz}
\usepackage{wrapfig}
\usepackage{arydshln}
\usepackage{xcolor}

\usepackage{titlesec}
\usepackage{titlecaps}
\newtheorem{theorem}{Theorem}[section]

\newtheorem{lemma}[theorem]{Lemma}

\newtheorem{proposition}[theorem]{Proposition}

\graphicspath{ {./images/} }
\usepackage{hyperref}
\usetikzlibrary{positioning}

\newcommand{\R}{{\mathbb R}}

\newcommand{\C}{{\mathbb C}}

\newcommand{\T}{\mathcal{T}}
\newcommand{\Tw}{\widetilde{T}}

\newcommand{\CR}{\mathcal{R}}
\newcommand{\PC}{\mathcal{P}_j[\mathbb{C}]}

\newcommand{\Hp}{\mathcal{H}}

\usepackage{adjustbox}
\newcommand{\e}{\mathbf{e}}
\newcommand{\rv}{\mathbf{r}}
\newcommand{\lv}{\mathbf{l}}

\newcommand{\blambda}{\boldsymbol{\lambda}}
\newcommand{\abs}[1]{\left| #1 \right|}
\newcommand{\av}{\mathbf{a}}
\newcommand{\bv}{\mathbf{b}}

\newcommand{\norm}[1]{\left \| #1 \right \|}

\DeclareMathOperator{\diag}{diag}
\DeclareMathOperator{\re}{Re}
\DeclareMathOperator{\im}{Im}
\DeclareMathOperator{\adj}{adj}
\DeclareMathOperator{\GL}{GL}
\DeclareMathOperator{\Tr}{Tr}

\renewcommand{\arg}{\mathrm{arg}}

\renewcommand{\epsilon}{\varepsilon}

\newcommand{\be}{\begin{equation}}
\newcommand{\ee}{\end{equation}}
\numberwithin{equation}{section}
\usepackage{bbm}
\usepackage{nicefrac}
\usepackage{mathtools}
\usepackage{soul}

\begin{document}

\title{A matrix model of a non-Hermitian $\beta$-ensemble}
\author{Francesco Mezzadri\footnote{E-mail: f.mezzadri@bristol.ac.uk}}
\author{Henry Taylor\footnote{E-mail: ht17630@bristol.ac.uk}}
\affil{\small{\emph{School of Mathematics, University of Bristol, Fry Building, 
Woodland Road, Bristol BS8 1UG, UK}}}
\date{}
\maketitle
\begin{abstract}
    We introduce the first random matrix model of a complex $\beta$-ensemble.  
    The matrices are tridiagonal and can be thought of as the non-Hermitian
    analogue of the Hermite $\beta$-ensembles 
    discovered by Dumitriu and Edelman (\emph{J. Math. Phys.,} 
    \textbf{43} (2002), 5830). 
    The main feature of the model is that the exponent $\beta$ of the Vandermonde 
    determinant in the joint probability 
    density function (\emph{j.p.d.f.}) of the eigenvalues can take 
    any value in $\mathbb{R}_+.$ However, when $\beta=2$, the \emph{j.p.d.f.}
    does not reduce to that of the Ginibre ensemble, but it contains an extra 
    factor expressed as a multidimensional integral over the space of the eigenvectors.
\end{abstract}
\section{Introduction}
\label{introduction}
In 1965 Ginibre~\cite{Gin65} computed the the joint probability 
density function  (\emph{j.p.d.f.}) of the eigenvalues of non-Hermitian matrices whose elements are independent  
complex normal random variables, and showed that
\begin{equation}
    \label{complex_ginibre}
    P_2(\blambda) = \frac{1}{Z_2} \prod_{j=1}^n
    \exp\left(-\abs{\lambda_j}^2\right)\prod_{1\le j < k \le n}\abs{\lambda_k - \lambda_j}^\beta,
\end{equation}
where $\blambda = (\lambda_1,\dotsc,\lambda_n)\in \mathbb{C}^n$ is the spectrum, $\beta=2$, $n$ is the dimension of the matrices and $Z_2$ is a normalisation constant.
This ensemble of random matrices  is known  in the literature 
as \emph{Ginibre Unitary Ensemble,} or GinUE.  

Formula~\eqref{complex_ginibre} can be interpreted as the 
equilibrium distribution of  a two-dimensional system of unit point charges, 
which interact with logarithmic repulsion in a confining quadratic
potential at inverse temperature $\beta=2$. This is known as  
\emph{Dyson's Coulomb gas model.} 
The realization of a two-dimensional gas of charged particles in terms of eigenvalues  of matrices in the GinUE means that 
at $\beta=2$ the model is exactly integrable, and can be studied analytically in great detail.  Recent advances in the theory
of non-Hermitian random matrices have prompted a large body of research into
this statistical mechanical model.  Non-Hermitian random matrices
have found applications in quantum chromodynamics, growth problems, scattering theory,
open quantum systems and quantum chaos. (See the review articles~\cite{BF25,FS03} and references therein.) 

Given the success of Dyson's Coulomb gas model, it is natural to ask
whether it is possible to create a \emph{complex} $\beta$-\emph{matrix model,} 
analogous to the tridiagonal $\beta$-Hermite ensemble of Hermitian matrices introduced 
by Trotter~\cite{Tro84} when $\beta=1,2,4$ and extended to
$\beta \in \mathbb{R}_+$ by Dumitriu and Edelman~\cite{DE02}.
The purpose of this article is to introduce the first non-Hermitian $\beta$-ensemble. 

Consider the complex tridiagonal matrix
\begin{equation}
\label{eq:tridiagmatrix}
T = \begin{pmatrix}
a_n & 1 &  &  &  \\
b_{n-1} & a_{n-1} & 1 &  &  \\
 & \ddots & \ddots & \ddots &  \\
 & & b_2 & a_2 & 1 \\
 &  &  & b_1 & a_1 
\end{pmatrix}.
\end{equation}
Let $\mathcal{N}(\mu,\sigma)$ denotes 
the normal distribution with mean $\mu$ and standard deviation 
$\sigma$, $\mathcal{U}(a,b)$ the uniform distribution in $[a,b)$ and 
$\chi_\alpha$ the chi-distribution with parameter $\alpha \in 
\mathbb{R}_+$.  Recall that the probability density 
function of a chi-distributed random variable is 
\[
f_{\alpha}(x) =\begin{cases}
\frac{2^{1-\alpha/2}}{\Gamma\left(\frac{\alpha}{2}\right)}e^{-x^2/2}x^{\alpha-1}&\text{for}\hspace{2mm}x\geq0, \; \alpha >0, \\
0&\text{otherwise}.
\end{cases}
\]The elements of $T$ are independent random variables with distributions
\begin{subequations}
  \label{mat_el_dis} 
\begin{gather}
    \re(a_j) \sim  \mathcal{N}(0,1),  \quad \im(a_j) \sim \mathcal{N}(0,1), \quad j=1,\dotsc,n,  \\
    \abs{b_j} \sim \chi_{\beta j/2}, \quad j=1,\dotsc,n-1, \\
    \arg (b_j) \sim \mathcal{U}(0,2\pi), \quad j=1,\dotsc,n-1,
\end{gather}
\end{subequations}
where $\abs{b_j}$ and $\arg(b_j)$ are independent for each $j=1,\dotsc,n-1$. Let $\blambda = (\lambda_1,\dotsc\lambda_n)\in \C^n.$ The main result of this article is the following theorem.
\begin{theorem}
\label{main_theorem}
The j.p.d.f. of the eigenvalues of the random matrix~\eqref{eq:tridiagmatrix} is
\begin{equation}
\label{eq:EigDens1}
    P_\beta(\blambda) = \frac{1}{Z_\beta} \prod_{j=1}^n
    \exp\left(-\frac{\abs{\lambda_j}^2}{2}\right)\prod_{1\le j < k \le n}\abs{\lambda_k - 
    \lambda_j}^\beta f(\blambda),
\end{equation}
where $\beta\in \mathbb{R}_+$, $Z_\beta$ is a normalization constant,
\begin{equation}
\label{factor}
f(\blambda) =\int_{\mathbb{C}^{n-1}}
\exp\big(-g(\blambda,\mathbf{r})\big)
\prod_{j=1}^{n}|r_j|^{\frac{\beta}{2}-2}\prod_{k=1}^{n-1}d^2r_k,
\end{equation}
$\mathbf{r} =(r_1,\dotsc,r_{n}) \in \mathbb{C}^{n}$, $r_n = -r_1 - \dotsb - r_{n-1}$ and
$d^2r_j=\tfrac{i}{2}dr_j \wedge d\bar{r}_j$.  The exponent $g(\blambda,\mathbf{r})$ is 
a positive function defined by 
\begin{equation}
g(\blambda,\mathbf{r}) =  \frac{1}{2}\left(\sum_{j=1}^n \abs{a_j}^2 + \sum_{j=1}^{n-1}\abs{b_j}^2 
        + n - 1\right) - \frac{1}{2}\sum_{j =1}^n \abs{\lambda_j}^2,
\end{equation}
where the $a_j$s and $b_k$s are rational functions of $\lambda_1,\dotsc,\lambda_n$ and $r_1,\dotsc,r_{n-1}$,
and can be reconstructed interactively from the recurrence relations~\eqref{eq:sol_rec_rn}. 
(See Theorem~{\rm \ref{theorem:inversion}}.) Furthermore, we have
\begin{enumerate}
\item  $g(\blambda,\mathbf{r}) = O\left( \abs{\lambda_j}^4\right)$ as $\lambda_j \to \infty$, \quad $j=1,\dotsc,n$;
\item $g(\blambda e^{i\psi},\rv)=g(\blambda,\rv)$, \quad   
 $\,\psi \in [0,2\pi)$.
\end{enumerate}
\end{theorem}
It follows from Property 2 is that the density of 
the eigenvalues is invariant under rotations, see Fig.~\ref{fig:image2}. 
 This is also the case for the GinUE, but while in the Ginibre ensemble we can integrate 
over the eigenvectors, in this model even when $\beta=2$ such  integration is not 
straightforward. The factor $f(\blambda)$ can only be 
studied through its integral representation~\eqref{factor}. As a consequence,  the limiting local behaviour of the eigenvalues is given by the determinantal point process arising as scaling limit of the Ginibre ensemble.
In general, 
for large matrix dimension the local statistics of the eigenvalues is governed 
by 
\[ 
\abs{\Delta(\blambda)}^\beta = \prod_{1\le j < k \le n}\abs{\lambda_k - 
    \lambda_j}^\beta 
\]
and the factor $f(\blambda)$ in the \emph{j.p.d.f.} of the 
eigenvalues~\eqref{eq:EigDens1}.  The function $g(\blambda,\mathbf{r})$ is positive-definite, see~Eq.~\eqref{exponent}, and the integral on the right-hand side of Eq.~\eqref{factor} does not vanish when the $\lambda_j$s coincide pairwise; therefore, we would expect that probability density of the eigenvalue distances 
$s = \abs{\lambda_j - \lambda_k}$ should be  of $O\left(s^{\beta + 1}\right)$ for small $s.$
\begin{figure}[ht]
\begin{subfigure}{0.5\textwidth}
\includegraphics[width=1\linewidth, height=5cm]{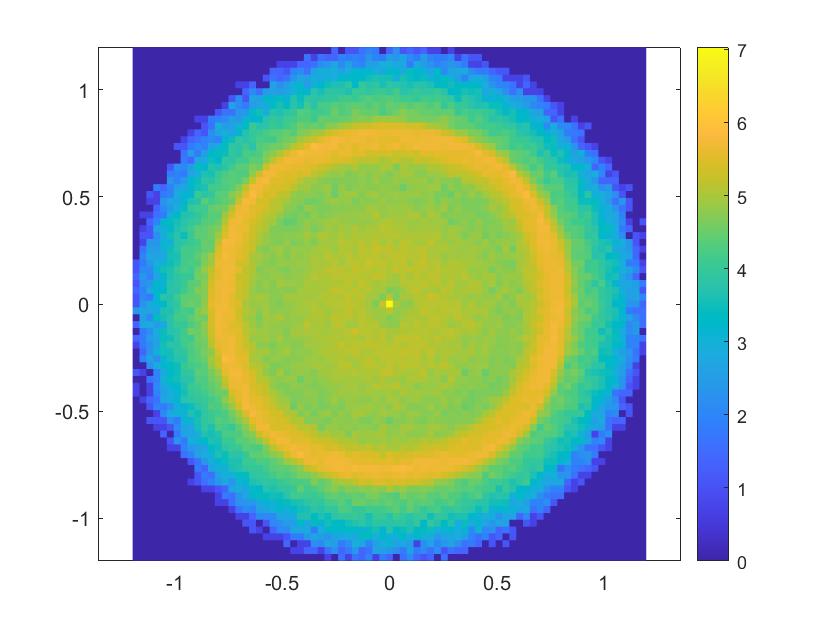} 
\caption{$\beta=2$ and $n=5000$}
\end{subfigure}
\begin{subfigure}{0.5\textwidth}
\includegraphics[width=1\linewidth, height=5cm]{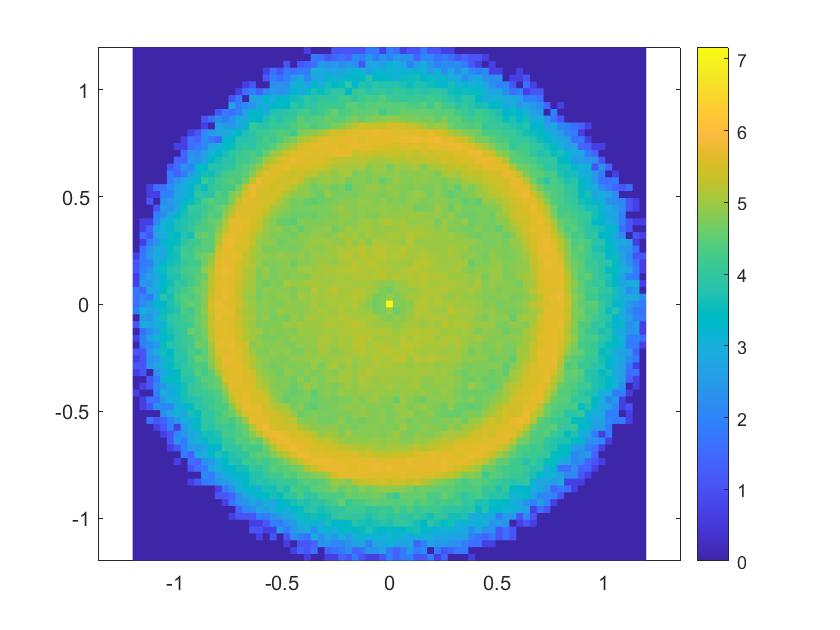}
\caption{$\beta=6$ and $n=5000$}
\end{subfigure}
\caption{Density-of-states of the complex eigenvalue spectrum of $(\ref{eq:tridiagmatrix})$ for 250 matrices of varying 
$\beta$ and dimension $5000\times 5000$, whose elements are distributed according to $(\ref{mat_el_dis})$. The density is normalized by $\sqrt{n \beta /4}$.}
\label{fig:image2}
\end{figure}

The matrices~\eqref{eq:tridiagmatrix} can be thought of as the complex analogues
of the tridiagonal model discovered by Dumitriu and Edelman~\cite{DE02}, the
$\beta$-Hermite ensemble. An important difference, however,  is that 
an arbitrary matrix in $\GL(n,\mathbb{C})$ cannot be 
reduced to the form~\eqref{eq:tridiagmatrix} by a similarity transformation, 
while a Hermitian matrix can always be mapped to a tridiagonal matrix by applying an appropriate
sequence of Householder reflections.   Indeed, Jacobi matrices are universal models of
self-adjoint operators. If the matrices~\eqref{eq:tridiagmatrix} were normal, their Schur factorization would
not contain a strictly upper-diagonal matrix, see~\eqref{schur} and~\eqref{exponent}; therefore, 
$g(\blambda ,\mathbf{r})=0$ and $f(\blambda)=1$. However,  in
general, the tridiagonal matrices~\eqref{eq:tridiagmatrix} are not 
normal and we would not expect to recover the \emph{j.p.d.f.}~\eqref{complex_ginibre} in
a model of tridiagonal complex matrices.

In a pioneering article Krishnapur \emph{et al.}~\cite{KRV16} introduced a novel
approach to study universality in random matrix theory.  Their idea is based on the
observation that after scaling, a Jacobi matrix belonging to the $\beta$-Hermite
ensemble converges as $n \to \infty$ to the stochastic Airy operator. It is not clear if similar
approach can be used for  the matrices~\eqref{eq:tridiagmatrix} to prove
spectral universality in non-Hermitian ensembles.   The reason is that, while a Jacobi matrix is uniquely associated to a set of orthogonal polynomials on the real line, there is no
relation between the matrix~\eqref{eq:tridiagmatrix} and orthogonal polynomials in the complex plane. 

The classical theory of random matrices developed in the 1950s and 1960s by
Wigner, Dyson and others~(see, \emph{e.g.}, \cite{For10,Meh04} 
for a discussion) is centred around three fundamental classes of ensembles, which correspond to the three types of  
associative division algebra with real coefficients, namely the fields of the real, complex  and 
quaternion numbers. This classification is known as   
\emph{Dyson's threefold way}~\cite{Dys62e}. 
Typical examples are the Gaussian Orthogonal Ensemble (GOE), 
Gaussian Unitary Ensemble (GUE) and the Gaussian 
Symplectic Ensemble (GSE); they are composed of real symmetric, complex Hermitian and self-dual real 
quaternion matrices, respectively, whose independent elements are normal random variables. The \emph{j.p.d.f.} 
of the eigenvalues is given by
\begin{equation}
    \label{gaussian_jpdf}
     P_\beta(\blambda) = \frac{1}{Z_\beta}\prod_{j=1}^n \exp\left(-\frac{\beta}{2}
     \lambda_j^2\right) \prod_{1\le j < k\le n}\abs{\lambda_k-\lambda_j}^\beta,
\end{equation}
where $n$ is the dimension of the matrices and $\lambda_j\in  \R$, $j=1,\dotsc,n$. The parameter $\beta$ is 
the dimension of the division algebra: $\beta=1$ for the 
GOE; $\beta=2$  for the GUE;  $\beta=4$ for the GSE. 
Dyson~\cite{Dys62a,Dys62b,Dys62c} extended this theory to unitary matrices and
introduced the Circular Unitary Ensemble (CUE),  
Circular Orthogonal Ensemble (COE) and Circular Symplectic Ensemble (CSE), 
which fall into the same classification scheme.
Zirnbauer~\cite{Zir96}, Altland and Zirnbauer~\cite{AZ97} and 
later Due\~{n}ez~\cite{Due01,Due04} generalised 
Dyson's threefold 
way to include all Cartan's symmetric spaces. 

In a pioneering article, Dumitriu  and Edelman~\cite{DE02} discovered a new class of 
random matrix models, the $\beta$-Hermite ensemble, which are made of 
tridiagonal matrices whose \emph{j.p.d.f.} of the  eigenvalues is formally the same 
as that in Eq.~\eqref{gaussian_jpdf}, but with the crucial feature that the exponent 
$\beta$ of the Vandermonde determinant can be any positive real number and does not 
identify a particular symmetric space. They also introduced the $\beta$-Laguerre  
ensembles. In 2004 Killip and Nenciu~\cite{KN04} extended the theory of Edelman and Dumitriu 
to matrix models of unitary matrices and generalised Dyson's circular ensembles 
to any $\beta\in \mathbb{R}_+$.  They also discovered the $\beta$-Jacobi
ensembles. 
Dumitriu and Forrester~\cite{DF10} gave a tridiagonal realisation 
of the antisymmetric Gaussian $\beta$-ensemble. Forrester and Rains~\cite{FR05} 
studied one-parameter interpolation of classical matrix ensembles. The impact of the
work of Dumitriu and Edelman~\cite{DE02} cannot be understated; 
it opened new areas of research in random matrix theory~\cite{AMP22,BEY12,BEY14}, the 
theory of stochastic processes~\cite{ES07,KRV16,MPS20}, Anderson 
localisation~\cite{BFS07}, integrable systems~\cite{GGGM23,GM23,Spo20,Spo21}, 
statistical mechanics~\cite{AS21,For21,NS15,SS12,SS15}.

The article is structured as follows: in Sec.~\ref{tridiag_mat} we introduce
the model and establish the map between the spectral parameters and 
the matrix~\eqref{eq:tridiagmatrix}; Sec.~\ref{sec:Jacobian} is dedicated
to the computation of the Jacobian of this map; 
in Sec.~\ref{sec:MainThm} we complete the proof of Theorem~\ref{main_theorem}.

\section{Spectral decomposition of tridiagonal matrices}
\label{tridiag_mat}
In order to prove Theorem~\ref{main_theorem}, we need to establish a few important
properties of non-Hermitian tridiagonal matrices.  Consider the matrix 
\be
\label{eq:gen_trid}
\widetilde{T}= \begin{pmatrix}
a_n & c_{n-1} &  &    \\
\widetilde{b}_{n-1} & a_{n-1} & c_{n-2}   &  \\
 & \ddots & \ddots &   \\
 &  &  \widetilde{b}_1 & a_1 
\end{pmatrix},
\ee
with $a_j,\widetilde{b}_j,c_j\in\mathbb{C}$ and $c_j \neq 0$, $j=1,\dotsc,n$. 
Then, $T = C \widetilde{T} C^{-1}$, where $b_j = c_j \widetilde{b}_j$ and
\[
C = \begin{pmatrix} 1 & &  & \\
                      & c_{n-1} & & \\
                      & & \ddots &\\
                      &  & & c_{n-1}\dotsm c_1
                      \end{pmatrix}.
\]
The matrices $T$ and $\widetilde T$ have the same spectrum, but $T$ has a smaller number of degrees of freedom, which makes the analysis of its spectrum  more amenable.   In Section~3 we will compute the Jacobian between the matrix elements of $T$ and its spectral parameters.  
It is well known that  Jacobi matrices are uniquely determined by their eigenvaules and the first row of the matrix of the eigenvectors (see, e.g., Parlett~\cite{Par98}).  The calculation of the Jacobian of $T$ is based on the observation that this property is true also for the tridiagonal matrix~\eqref{eq:tridiagmatrix}.

Denote by $\T(n)$ the set of matrices $T$ of the form~\eqref{eq:tridiagmatrix} such that 
$\det(T)\neq 0$ and whose spectrum is simple.\footnote{While the spectra of Jacobi matrices are non-degenerate, the 
matrix~\eqref{eq:tridiagmatrix} may have multiple eigenvalues. 
See Appendix~\ref{multiplicity}  for a discussion about tridiagonal complex matrices with 
degenerate spectra.}  The probability distributions~\eqref{mat_el_dis} define 
an absolutely continuous measure on the set of tridiagonal matrices~\eqref{eq:tridiagmatrix}; therefore, 
the exceptional set of matrices with degenerate spectrum forms a set of measure zero. We define
\begin{align*}
        D(n)  & = \left\{\diag(z_1,\dotsc,z_n)\colon z_j \in \C, \quad  z_j \neq 0, \quad
         j=1,\dotsc,n \right\}, \\
        \Lambda(n) & =  \left\{ \diag(\lambda_1,\dotsc,\lambda_n)\colon \lambda_j,\lambda_k \in \C, 
        \quad \lambda_j \neq \lambda_k \neq 0, \quad 1 \le j < k \le n \right\},\\
        R(n)  & = \left\{R \in \GL(n,\C) \colon T = 
        R\Lambda R^{-1},\quad T \in \T(n),\quad  \Lambda \in \Lambda(n)\right\},
\end{align*}
where $\diag\left( \cdot \right)$ indicates a diagonal matrix. Let $T \in \T(n)$.  The spectral decomposition
\begin{equation}
    \label{diag_fact}
    T = R\Lambda R^{-1}, 
\end{equation}
where $\Lambda \in \Lambda(n)$ and  $R \in R(n)$, is unique up to a permutation of the eigenvalues and 
right multiplication $R \mapsto RD$, $D\in D(n)$. Hence, the factorization~\eqref{diag_fact} defines a 
bijection
\begin{equation}
\label{sing_valued_map}
    \mathcal{F} \colon \T(n) \to  \mathcal{L}(n) \times \CR(n),
\end{equation}
where
\[
\CR(n)= R(n)/D(n),  \quad  \mathcal{L}(n) = \Lambda(n)/\mathfrak{S}_n,
\]
and $\mathfrak{S}_n$ is the symmetric group of order $n$.  
We will also study the restriction of 
$\mathcal{F}$ to $\T'(n) = \T(n) \setminus B$, where
\begin{equation}
    \label{mat_b_zero}
        B= \bigl\{ T \in \T(n) \colon b_j = 0 \:\: \text{for some} \:\: j \in \{1,\dotsc,n-1\} \bigr\}.
\end{equation}
Matrices such that one or more of the $b_j$s are zero are discussed in
Appendix~B. Write
\begin{gather*}
    R'(n)  = \left\{R \in \GL(n,\C) \colon T = 
    R\Lambda R^{-1},\quad T \in \T'(n),\quad  \Lambda \in \Lambda(n)\right\},\\
    \CR'(n)= R'(n)/D(n).\\   
\end{gather*}
\begin{lemma}
\label{parameterization}
Let $\rv = (r_1,\dots,r_n)$ and $\mathbf{v}=(v_1,\dotsc,v_n)$ be the first 
and last row of  $R\in R'(n)$, respectively. We have 
\begin{equation}
    \label{el_diff_from_zero}
        r_j \neq 0, \quad  v_j \neq 0 \quad \text{for all}  \quad j \in \left \{1,\dotsc,n\right\}.
\end{equation}
Furthermore, the subset of $R'(n)$  such that 
\begin{equation}
    \label{hyperp_1}
        r_1 + \dotsb + r_n =1
\end{equation}
spans a set of representatives in $\CR'(n)$.  
\end{lemma}
\begin{proof}
Let us take  $T\in \T'(n)$ and let
\begin{equation}
    \label{princ_subm}
        T_k = \begin{pmatrix}
                 a_n & 1 &  &  \\
                 b_{n-1} & a_{n-1} & \ddots &  \\
                 & \ddots & \ddots & 1 \\
                 &  & b_{n-k+1} & a_{n-k+1} 
              \end{pmatrix}
\end{equation}
be the principal submatrix of $T$ obtained by removing the last $n-k$ rows and columns. Note that
$T_n=T$. Let $\lambda_j^{(k)}$, $j=1,\dotsc,k$, denote the eigenvalues of $T_k$ and write
\be
\label{cp_sm1}
    \Phi_k(x) = \det\left(x I - T_{k}\right) = 
                 \prod_{j=1}^k\Bigl(x-\lambda_j^{(k)}\Bigr).
\ee
The characteristic polynomials $\Phi_k$ obey the recurrence relation 
\begin{equation}
    \label{rec_rel1}
        \Phi_{k}(x) = (x - a_{n-k +1 })\Phi_{k-1}(x) - b_{n-k +1 }\Phi_{k-2}(x), \quad k=1,\dotsc, n,
\end{equation}
with initial conditions 
\begin{equation}
    \label{in_con}   
        \Phi_{0}(x)=1,  \quad \Phi_{-1}(x)=0
\end{equation}
and $b_n$ undetermined.  Equation~\eqref{cp_sm1} implies that $\lambda_j =\lambda_j^{(n)}$ is
an eigenvalue of $T$ and that
\begin{equation}
    \label{right_eig}
        \begin{pmatrix} 
            \Phi_0\bigl(\lambda_j\bigr)\\ 
            \vdots \\
            \Phi_{n-1}\bigl( \lambda_j\bigr) 
        \end{pmatrix}
\end{equation}
is the right-eigenvector corresponding to $\lambda_j$.  Furthermore, by 
definition, the matrix 
\begin{equation*}
    R= \begin{pmatrix}
            \Phi_0\bigl(\lambda_1\bigr) & \cdots & \Phi_0\bigl(\lambda_n\bigr)\\
            \Phi_1\bigl(\lambda_1\bigr) & \cdots & \Phi_1\bigl(\lambda_n\bigr)\\
             \hdotsfor[2]{3}\\
            \Phi_{n-1}\bigl(\lambda_1\bigr)& \cdots & 
             \Phi_{n-1}\bigl(\lambda_n\bigr)
        \end{pmatrix}
\end{equation*}
belongs to $R'(n)$. If there is no degeneracy in the spectrum, the eigenvectors are unique
up to a multiplication by a constant.  Since $\Phi_0\bigl(\lambda_j\bigr)=1,$ it follows that
$r_j\neq 0$ for $j=1,\dotsc,n.$

In order to show that $v_j\neq 0$, suppose that 
\[
    \Phi_{n-1}\bigl(\lambda_j\bigr)= 0 \quad  \text{for some} \quad 
    j \in \{1,\dotsc,n\}.
\]
Since $b_{j}\neq 0$, $j=1,\dotsc,n-1$, it  follows by induction that all the entries 
of~\eqref{right_eig} are zero and that $\det R = 0$, which contradicts the assumption that the 
spectrum is simple. 

We now want to show that
\begin{equation}
    \label{setrep}
        L = \left\{ R \in R'(n) \colon  r_1 + \dotsb + r_n =1\right\}
\end{equation}
is a set of representatives of $R'(n)$. Suppose that the elements of the first column of $R^{-1}$ are all 
different from zero. Then,  Eq.~\eqref{hyperp_1} follows from the identity $RR^{-1}=I$ and from the fact that 
matrices in the same equivalent class of $\CR'(n)$ differ by the right multiplication  $R \mapsto RD$, where
$D\in D(n).$ The proof that the elements of the first column of $R^{-1}$ are all different from zero 
follows similar lines to those that we have used to prove 
Eq.~\eqref{el_diff_from_zero}.

Take the principal submatrix of $T$ obtained by removing the first $n-k$ rows and columns,
\begin{equation}
    \label{princ_subm_2}
        U_k = \begin{pmatrix}
                a_k & 1 &  &  \\
                b_{k-1} & a_{k-1} & \ddots &  \\
                 & \ddots & \ddots & 1 \\
                 &  & b_1 & a_1 
              \end{pmatrix}.
\end{equation}
Let $\mu_j^{(k)}$, $j=1,\dotsc,k$, denote the eigenvalues of $U_k$ and write
\be
    \label{cp_sm2}
        \chi_k(x) = \det\left(x I - U_{k}\right) = 
        \prod_{j=1}^k\Bigl(x-\mu_j^{(k)}\Bigr).
\ee
We have that $\lambda_{j}= \mu^{(n)}_{j}.$ The polynomials $\chi_k$ obey the recurrence relation 
\begin{equation}
\label{rec_rel2}
\chi_k(x) = (x - a_k)\chi_{k-1}(x) -b_{k-1}\chi_{k-2}(x), \quad k=1,\dotsc,n,
\end{equation}
with initial conditions 
\[
\chi_{0}(x)=1, \quad  \chi_{-1}(x)=0
\]
and $b_0$ undetermined.  Equation~\eqref{rec_rel2} and
\[
    \chi_n\bigl(\lambda_j\bigr)=0, \quad j=1,\dotsc,n,
\]
imply that 
\begin{equation}
    \label{left_eig}
        \left(\chi_{n-1}\bigl(\lambda_j\bigr),\dotsc, \chi_0\bigl( \lambda_j\bigr) \right)
\end{equation}
is the left eigenvector corresponding to $\lambda_j$ and is proportional to the $j$-th row of $R^{-1}.$ 
The same arguments that led to the conclusion that $\Phi_{n-1}\bigl(\lambda_j\bigr)\neq 0$ give
$\chi_{n-1}\bigl(\lambda_j\bigr)\neq 0.$
\end{proof}

Define the set
\begin{equation}
    \label{hyperplane}
        \Hp_n = \left\{ \rv \in \C^n \colon r_1 + \dotsb + r_n =1, \quad r_j \neq 0, 
        \quad j=1,\dotsc,n \right\}.
\end{equation}
We want to show that~\eqref{sing_valued_map} induces an injective map 
\begin{equation}
    \label{induced_map}
        \mathcal{G} \colon \T'(n) \to \mathcal{L}(n) \times \Hp_n.
\end{equation}
Lemma~\ref{parameterization} and Eq.~\eqref{sing_valued_map} imply that $\mathcal{G}$ is single-valued. 
In order to prove that it is injective, we need some more work. In Sec.~\ref{sec:Jacobian} we will compute 
the Jacobian of $\mathcal{G}.$ Let us denote 
\begin{equation}
    \label{eq:newnot}
        R =\begin{pmatrix}\rv_1^t\\ \vdots \\ \rv_n^t\end{pmatrix}, \quad 
            R^{-1}= \left(\lv_1,\dotsc,\lv_n\right),
\end{equation}
where 
\[
\rv_j^t = \left( r_{j1},\dotsc, r_{jn}\right), \quad
 \lv_j = \begin{pmatrix} l_{1j} \\ \vdots \\ l_{nj} \end{pmatrix}, \quad j=1,\dotsc,n.
\]
In addition we will write
\begin{equation*}
\rv^t_j \lv_{k} = r_{j1}l_{1k}+ \dotsb + r_{jn}l_{nk} = \delta_{jk}, \quad 
j,k=1,\dotsc,n,
\end{equation*}
where the notation on the left-hand side is not the usual scalar product in $\C^n$, but
stands for matrix multiplication between
the row vector $\rv_j^t$ and the column vector $\lv_k$.

\begin{theorem}
\label{theorem:inversion}
Let $\Lambda \in \Lambda(n)$ and $\rv = (r_1,\dotsc,r_n) \in \Hp_n$. There exist a unique 
$R\in R'(n)$ whose first row is $\rv$ and a unique $T\in \T'(n)$  
such that $T = R\Lambda R^{-1}$, except for $\rv$ belonging to an exceptional set $M\subset \Hp_n$
of measure zero.
\end{theorem}
\begin{proof}
Take $R$ from the set of representative of $\CR'(n)$ defined in Lemma~\ref{parameterization}. Let us 
write the matrix equations
\[
TR = R\Lambda \quad \text{and} \quad  R^{-1}T = \Lambda R^{-1}
\]
as a set of vector equations:
\begin{subequations}
\label{eq:rec_r}
\begin{align}
\label{eq:rec_ra}
\rv_j^t\Lambda  & = b_{n-j+1}\rv_{j-1}^t + a_{n-j+1}\rv_j^t + \rv_{j+ 1}^t, \\
\label{eq:rec_rb}
\Lambda\lv_j & = \lv_{j-1} + a_{n-j+1}\lv_j +b_{n-j}\lv_{j+1},
\end{align}
\end{subequations}
for $j=1,\dotsc,n$, where we have adopted the notation~\eqref{eq:newnot}.  Let
$\rv_1=\rv$ and

We want to show that given $\Lambda$ and $\rv_1^t$ with $\rv^t_1\lv_1=1$, where 
\[
    \lv_1 = \begin{pmatrix}
        1 \\ \vdots \\ 1
    \end{pmatrix},
\]
we can reconstruct $T$, $R$ and $R^{-1}$ uniquely from the recurrence relations~\eqref{eq:rec_r} with boundary conditions $b_n=0$ and $\lv_0=\mathbf{0}$,
$\rv_{n+1}^t=\mathbf{0}$. The vectors $\rv_0^t$ and $\lv_{n+1}$ are undetermined.

Write
\begin{subequations}
\label{eq:sol_rec_r1}
\begin{align}
\label{eq:an}
a_n & = \rv_1^t \Lambda \lv_1,\\
\label{eq:r2}
\rv_2^t & = \rv_1^t \Lambda - a_n \rv_1^t,\\
\label{eq:bn1}
   b_{n-1} & = \rv_2^t \Lambda   \lv_1,\\
   \label{eq:l2}
        \lv_2 &= \frac{1}{b_{n-1}}\left(\Lambda  \lv_1 - a_n \lv_1\right).
\end{align}
\end{subequations}
For any fixed $\Lambda,$ the set exceptional points $\rv_1$ such that $b_{n-1}=0$ has measure zero 
in $\Hp_n$ and can be neglected. The quantities $a_n$, $b_{n-1}$, $\rv_1^t$ and $\lv_1$ solve~\eqref{eq:rec_r}
for $j=1$. We need to check that they are consistent with the 
condition $RR^{-1} = I$. Eqs.\eqref{eq:an} and~\eqref{eq:r2} give $\rv_2^t  \lv_1=0$. Similarly, by 
Eqs.~\eqref{eq:an} and~\eqref{eq:l2} we have $\rv_1^t  \lv_2 =0$.  Finally, we have
\begin{equation*}
    \rv_2^t  \lv_2 =  \frac{1}{b_{n-1}} \left(\rv_2^t \Lambda 
    \lv_1 - a_n\rv_2^t  \lv_1\right) = 1.
\end{equation*}
Note that combining the orthogonality relations $\rv_1^t \lv_2 = 0$ and $\rv_2^t \lv_2=1$ with Eq.~\eqref{eq:r2} gives
\[
\rv_1^t  \Lambda  \lv_2 = 1,
\]
as it should be. 

When $1<j<n$ take 
\[
\rv_{j-1}^t,\; \rv_j^t,\; \lv_{j- 1},\; \lv_j
\]
subject to the conditions
\begin{subequations}
\label{eq:int_ic}
\begin{gather}
\rv^t_j \lv_j= \rv^t_{j-1}\lv_{j-1}=1, \quad \rv^t_j \lv_{j-1}= \rv^t_{j-1} \lv_{j}=0,\\
b_{n - j +1} = \rv_{j}^t  \Lambda \lv_{j-1} \neq 0, \quad \rv_{j-1}^t \Lambda \lv_{j}=1.
\end{gather}
\end{subequations}
Define
\begin{subequations}
\label{eq:sol_rec_rn}
\begin{align}
\label{eq:anj}
a_{n-j+1} & = \rv_j^t \Lambda  \lv_j,\\
\label{eq:rj}
\rv_{j+1}^t & =  \rv_j^t \Lambda - a_{n - j + 1}\rv_j^t - b_{n-j+1}\rv_{j-1}^t\\
\label{eq:bnj}
   b_{n-j} & = \rv_{j+1}^t  \Lambda   \lv_j,\\
   \label{eq:lj}
        \lv_{j+1} &= \frac{1}{b_{n-j}}\left(\Lambda  \lv_j- a_{n - j +1}\lv_j-\lv_{j-1}\right).
\end{align}
\end{subequations}
For any fixed $\Lambda$, we have $b_{n-j}=0$ for $\rv_{j+1}$ and $\lv_j$ belonging to a set of measure zero in $\C^n.$ 
In turn, $\rv_{j+1}$ and $\lv_j$ are rational functions of the elements of $\rv_1$; therefore,  $b_{n-j}=0$ only for 
$\rv_1$ belonging to a set of  measure zero in $\Hp_n.$ Combining~\eqref{eq:int_ic} with~\eqref{eq:sol_rec_rn} and
proceeding in a similar way as when $j=1$, we can show that
\begin{align*}
\rv_{j+1}^t \lv_{j} &= 0, & \rv_{j+1}^t  \lv_{j-1}& =0,\\
\rv_{j}^t \lv_{j+1} &= 0, & \rv_{j-1}^t  \lv_{j+1}&=0.
\end{align*}
Furthermore, we have
\begin{align*}
\rv_{j+1}^t \lv_{j+1}& = \frac{1}{b_{n-j}}\left(\rv_{j+1}^t 
\Lambda  \lv_j- a_{n - j +1}
\rv_{j+1}^t \lv_j-\rv_{j+1}^t \lv_{j-1}\right)=1\\
\rv_{j}^t   \Lambda   \lv_{j+1} & = b_{n-j+1}\rv_{j-1}^t  \lv_{j+1}
+ a_{n-j+1}\rv_j^t  \lv_{j+1}+ \rv^t_{j+ 1} \lv_{j+1} =1.
\end{align*}
The orthogonality relations
\[
\rv_{j+1}^t  \lv_{k}=0, \quad \rv_k^t  \lv_{j+1}=0, \quad k=1, \dotsc, j-2
\]
follow by induction and from the fact that by construction $\rv^t_j \Lambda \lv_k = 0$ if $\abs{j-k}>1$.

The boundary condition $\rv^t_{n+1}=\mathbf{0}$ terminates the recurrence relations.
\end{proof}

The proof of Theorem~\ref{theorem:inversion} breaks down when $T\in B,$ where the set $B$ was
introduced in~\eqref{mat_b_zero}.  Even if $B$ has measure zero and can be neglected, it is 
interesting to understand what happens when some of the matrix elements $b_j$ become zero.  The
adaptation of the proof  to include $T \in B$ is detailed in Appendix~\ref{zero_off_d}.

\section{The Jacobian}
\label{sec:Jacobian}
 
Let $\mathbf{z} =(z_1,\dotsc,z_n)\in \C^n$ and write $dz_j = d\re(z_j) + id\im(z_j)$.
We use the notation $d\mathbf{z} = \wedge_j \,dz_j,$ where the wedge product is taken over all
the components of $z_j,$ and $d^2 z_j = \tfrac{i}{2}dz_j \wedge d\bar{z}    _j.$

The following theorem generalises well known results for Jacobi 
matrices~\cite{DE02,FR06} to the  model~\eqref{eq:tridiagmatrix}. 
\begin{theorem}
\label{thm:jac}
We have
\begin{equation*}
\prod_{k=1}^{n-1}d^2b_k\prod_{j=1}^n d^2a_j=\frac{\prod_{j=1}^{n-1}\abs{b_j}^2}{\prod_{k=1}^n 
\abs{r_k}^2}\prod_{k=1}^{n-1}d^2r_k\prod_{j=1}^n d^2\lambda_j .
\end{equation*}
\end{theorem}
In order to prove this theorem, we need few preliminary results.  Denote by $\widetilde{T}_j$ the principal submatrix of~\eqref{eq:gen_trid} obtained by keeping the first $j$ 
rows and columns. We have $\widetilde T_n = \widetilde T$. Introduce the characteristic polynomial
\begin{equation}
       \label{cp_sm}
       \chi_j(z) =\det\left(zI - \widetilde{T}_j\right), \quad j=1,\dotsc,n,
\end{equation}
and set $\chi_0(z)=1$. Let  $\mathcal{P}_j[\mathbb{C}]$ be the set of monic
polynomials of degree $j$ over $\C$ and $\e_j$ be a canonical basis vector in $\C^n$, 
\emph{i.e.} $\e_j^t
=\overset{j}{(0,\dotsc,1,\dotsc,0)}$.
\begin{proposition}
\label{prop}
Let $j=1,\dotsc,n-1.$ The following statements hold:
\begin{enumerate}
    \item $\bigl|\widetilde{b}_{n-1}\widetilde{b}_{n-2}\dotsm \widetilde{b}_{n-j}\bigr|=
    \bigl \|\chi_j(\widetilde{T})\e_1\bigr \|=\underset{\psi \in \PC}
    {\min}\bigl \|\psi(\widetilde{T})\e_1\bigr \|$;
    \item $\chi_j(\widetilde{T})\e_1=\e_{j+1}\widetilde{b}_{n-1}\widetilde{b}_{n-2}\dotsm \widetilde{b}_{n-j}$;
    \item $|c_{n-1}c_{n-2}\dotsm c_{n -j}|=
    \bigl \|\e_1^t\chi_j(\widetilde{T})\bigr \|=\underset{\psi \in \PC}
    {\min}\bigl \|\e_1^t\psi(\widetilde{T})\bigr\|$;
    \item $\e_1^t\chi_j(\widetilde{T})= \e_{j+1}^tc_{n-1}c_{n-2}\dotsm c_{n-j}$.
\end{enumerate}
\end{proposition}
\begin{proof}
One can prove by induction that
\[
\e_k^t\widetilde{T}^j\mathbf{e}_1= \begin{cases} \widetilde{b}_{n-1}\widetilde{b}_{n-2}\dotsm \widetilde{b}_{n-j}, & \text{if $k=j+1$,}\\
0, & \text{if $k>j+1$.} \end{cases}  
\]
It follows that for any $\psi \in \PC$ we have
\[
\bigl \|\psi(\widetilde{T})\e_1\bigr \| \ge \bigl |\widetilde{b}_{n-1}\widetilde{b}_{n-2}\dotsm \widetilde{b}_{n-j}\bigr |. 
\]
This lower bound is sharp if there exist a $\psi\in \PC$ such that
\[
\psi(\widetilde{T})\e_1 = \e_{j+1}\widetilde{b}_{n-1}\widetilde{b}_{n-2}\dotsm \widetilde{b}_{n-j}.
\]

Note that the first $j$ components of $\widetilde{T}^j\e_1$ contain  only elements 
from $\widetilde{T}_j$, as do the vectors $\widetilde{T}^k\mathbf{e}_1$ for $1\le k<j$. 
Let 
\[
E_j=(\mathbf{e}_1,\dotsc ,\mathbf{e}_j).
\]
Then, we have
\[
E_j^t\psi(\widetilde{T})\e_1 = \psi(\widetilde{T}_j)\e_1,
\]
where, with an abuse of notation, $\e_1$ on the left-hand
side belongs to $\C^n$, while on the right-hand side $\e_1\in \C^j$.  
The Cayley-Hamilton theorem states that $\chi_j(\widetilde{T}_j)=0$. 
It follows that
\[
\chi_j(\widetilde{T})\e_1 = \e_{j+1}\widetilde{b}_{n-1}\widetilde{b}_{n-2}\dotsm \widetilde{b}_{n-j}.
\]
This proves statements 1. and 2.  The proof of 3. and 4. is analogous.
\end{proof}

Let 
\[
\Delta(\boldsymbol{\lambda})=\prod_{1\le j<k\le n}\left(\lambda_k-\lambda_j\right)
\]
denote the Vandermonde determinant.  
\begin{lemma}
\label{lem:jac1}
Let $T=R\Lambda R^{-1}$ be the spectral decomposition of $T\in \T(n)$, where $R$ is defined in~\eqref{eq:newnot} and
$\Lambda = \diag(\lambda_1, \dotsc,\lambda_n)$. We have
\[
\Delta(\boldsymbol{\lambda})^2=\frac{\prod_{j=1}^{n-1}b_j^j}{\prod_{j=1}^{n}r_{j}},
\]
where $\rv = (r_1,\dotsc,r_n)$ is the first row of $R$ and
\[
r_1 + \dotsb + r_n =1.
\]
\end{lemma}
\begin{proof}
The definition of $T$ and Proposition~\ref{prop} give
\begin{subequations}
\label{eq:minstat}
\begin{align}
\label{eq:lem231}
\chi_j(T)\mathbf{e}_1 & =\mathbf{e}_{j+1}b_{n-1}b_{n-2}\dotsm b_{n-j},\\
\label{eq:lem232}
    \mathbf{e}_1^t\chi_j(T) &= \mathbf{e}_{j+1}^t,
\end{align}
\end{subequations}
for $j=1,2,\dotsc,n-1$.  Using the spectral decomposition $T = R\Lambda R^{-1}$ and the definition of $R$ given in~\eqref{eq:newnot}, Eqs.~\eqref{eq:minstat} become
\begin{align*}
\chi_j(\Lambda)R^{-1}\mathbf{e}_1 & = \chi_j(\Lambda) \lv_1 = R^{-1}\mathbf{e}_{j+1}b_{n-1}b_{n-2}\dotsm b_{n-j} \nonumber \\
& = \lv_{j+1}b_{n-1}b_{n-2} \dotsm b_{n-j},\\
    \mathbf{e}_1^tR\chi_j(\Lambda) &= \rv_1^t\chi_j(\Lambda) = \mathbf{e}_{j+1}^t R  = \rv_{j+1}^t.
\end{align*}
Multiplying these equations side by side gives 
\[
\rv_1^t\chi_j(\Lambda)  \chi_k(\Lambda)\lv_1 = 
\rv_{j+1}^t  \lv_{k+1}\prod_{m=1}^k b_{n-m},
\]
which can be rewritten as 
\begin{equation}
\label{eq:lem235}
\sum_{l=1}^n r_l\chi_j(\lambda_l)\chi_k(\lambda_l)=
\begin{cases} 1 & \text{if $j=k=0$,} \\
\delta_{jk}\prod_{m=1}^k b_{n-m}, & \text{if $j \neq 0$ or $k \neq 0$,}
\end{cases}
\end{equation}
for $j,k=0,\dotsc, n-1$, where we have used $\lv_1^t = (1,\dotsc,1)$ and the fact that by 
definition $\chi_0(z)=1.$

Write
\begin{equation*}
    W = \begin{pmatrix} 1 & \chi_1(\lambda_1) & \dots &\chi_{n-1} (\lambda_{1}) \\
                        1 & \chi_1(\lambda_2) & \dots &\chi_{n-1} (\lambda_{2}) \\
                        \hdotsfor[2]{4}\\
                        1 & \chi_1(\lambda_n) & \dots &\chi_{n-1} (\lambda_{n}) \\
         \end{pmatrix}                        
\end{equation*}
and
\begin{equation*}
    Q = \begin{pmatrix} r_1 & & & \\
    & r_2 & &  \\
    & & \ddots& \\
    & &  & r_n 
    \end{pmatrix},  \quad 
    B = \begin{pmatrix} 1 & & & \\
    & b_{n-1} & &  \\
    & & \ddots& \\
    & &  & \prod_{m=1}^{n-1} b_{n-m} 
    \end{pmatrix}.
\end{equation*}
Eq.~\eqref{eq:lem235} is the matrix product
\begin{equation*}
    W^t Q W = B
\end{equation*}
written in terms of matrix elements.  Since $\det(W) = \Delta(\boldsymbol{\lambda})$, taking the determinant of both sides of this equation gives
\[
\prod_{1\le j<k\le n}\left(\lambda_k-\lambda_j\right)^2 
\prod_{j=1}^n r_j = \prod_{k=1}^{n-1} b_k^k.
\]
\end{proof}
A second proof of Lemma~\ref{lem:jac1} can be found in Appendix~\ref{sec_proof}. 

Let us denote
\begin{equation*}
        \av = (a_1,\dotsc,a_n) \in \mathbb{C}^n \quad \text{and} \quad \bv = (b_1,\dotsc,b_{n-1}) \in \mathbb{C}^{n-1},
\end{equation*}
where the $a_j$s and $b_j$s are the diagonal and off-diagonal elements of $T$, respectively.
\begin{lemma}
\label{lem:jac2}
 We have
\begin{equation}
\label{jac_abrlb}
  \prod_{j=1}^{n-1}b_{j}^{2j-1}  d\av \wedge d\bv =  
  \left(\prod_{k=1}^n r_k\right) \Delta(\blambda)^4 d\rv \wedge d\blambda.
\end{equation}
\end{lemma}
\begin{proof}
Write
\begin{equation*}
    ((I_n-zT)^{-1})_{11}=\sum_{j=1}^n\frac{r_j}{1-z\lambda_j}, 
\end{equation*}
and take the Taylor expansion of both sides of this equation. By equating the coefficients of the powers of $z$ up to $z^{2n-1}$ we arrive at
\begin{equation}
  \label{eq:Taylor_exps}
\begin{split}
    O(z^0):&\hspace{3mm}1=\sum_{j=1}^nr_j\\
    O(z^1):&\hspace{3mm}a_n=\sum_{j=1}^nr_j\lambda_j\\
    O(z^2):&\hspace{3mm}*\;+b_{n-1}=\sum_{j=1}^nr_j\lambda_j^2\\
    O(z^3):&\hspace{3mm}*\;+a_{n-1}b_{n-1}=\sum_{j=1}^nr_j\lambda_j^3\\
    O(z^4):&\hspace{3mm}*\; +b_{n-1}b_{n-2}=\sum_{j=1}^nr_j\lambda_j^4\\
    O(z^5):&\hspace{3mm}*\;+a_{n-2}b_{n-1}b_{n-2}=\sum_{j=1}^nr_j\lambda_j^5\\
    \vdots \hspace{.6cm} & \hspace{1.8cm}\vdots \hspace{2.5cm}\vdots\\
    O(z^{2n-1}):&\hspace{3mm}*\; +a_1b_1\dots b_{n-1}=\sum_{j=1}^nr_j\lambda_j^{2n-1}.
\end{split}
\end{equation}

The symbol  $*$ in Eqs.~\eqref{eq:Taylor_exps} 
indicates polynomials in which only the
$a_j$s and $b_k$s contained in lower order terms appear.
Therefore, if we choose the ordering
\[
a_n,  b_{n-1}, a_{n-1}, b_{n-2}, \dotsc, b_1,a_1,
\]
the matrix of the coefficients of the differentials of the left-hand sides of Eqs.~\eqref{eq:Taylor_exps} is triangular.  Taking the wedge product line by
line gives 
\begin{equation}
\label{eq:wplhs}
(-1)^{(n-1)(n-2)/2} \left(\prod_{j=1}^{n-1}b_j^{2j-1} \right)d\mathbf{a} \wedge d\mathbf{b}.
\end{equation}

Differentiating the right-hand sides of Eqs.~\eqref{eq:Taylor_exps} leads to
\begin{equation}
\label{eq:diffrhs}
d\left(\sum_{k=1}^nr_k\lambda_k^{j}\right)=\sum_{k=1}^{n-1}dr_k(\lambda_k^{j}-\lambda_n^{j})+j\sum_{k=1}^nr_k\lambda_k^{j-1}d\lambda_k,
\end{equation}
where we have used the relation
\begin{equation*}
    dr_n = - \sum_{j=1}^{n-1} dr_j,
\end{equation*}
which follows from the first equation in~\eqref{eq:Taylor_exps}. Taking 
the wedge product of the right-hand side of Eq.~\eqref{eq:diffrhs} 
for $j=1,\dotsc, 2n-1$ we arrive at 
\begin{equation}
\label{eq:jac:2}
    \left(\prod_{k=1}^nr_k \right)
    \det\bigg(\left(\lambda_k^j-\lambda_n^j\right)_{
    \begin{subarray}{l}
          j=1,\dotsc,2n-1 \\ k=1,\dotsc,n-1
     \end{subarray}}
     \left(j\lambda_k^{j-1}\right)_{
       \begin{subarray}{l}
          j=1,\dotsc,2n-1 \\ k=1,\dotsc,n
     \end{subarray}}
     \bigg)d \rv \wedge d\blambda. 
\end{equation}
The determinant in this equation can be calculated explicitly and is found as
Eq.~(1.175) in the book by Forrester~\cite{For10}.
We have
\begin{gather*}
(-1)^{(n-1)(n-2)/2}\det\bigg(\left(\lambda_k^j-\lambda_n^j\right)_{
    \begin{subarray}{l}
          j=1,\dotsc,2n-1 \\ k=1,\dotsc,n-1
     \end{subarray}}
     \left(j\lambda_k^{j-1}\right)_{
       \begin{subarray}{l}
          j=1,\dotsc,2n-1 \\ k=1,\dotsc,n
     \end{subarray}}
     \bigg) \nonumber \\
     \label{eq:crazy_det}
     = \prod_{1 \le j < k \le n}\left(\lambda_k - \lambda_j\right)^4 = \Delta(\blambda)^4.
\end{gather*}
Eqs.~\eqref{eq:wplhs} and~\eqref{eq:jac:2} give
\begin{equation*}
    \label{jacobiand_eq}
     \left(\prod_{j=1}^{n-1}b_j^{2j-1} \right)d\mathbf{a} \wedge d\mathbf{b} 
       =  \left(\prod_{k=1}^nr_k \right) \Delta(\blambda)^4 d\rv \wedge d\blambda,
\end{equation*}
which leads to~\eqref{jac_abrlb}.
\end{proof}
Combining Lemmas~\ref{lem:jac1} and~\ref{lem:jac2} proves Theorem~\ref{thm:jac}.

\section{Proof of Theorem \ref{main_theorem}}
\label{sec:MainThm}
The \emph{j.p.d.f.} of the matrix elements of matrices in $\T'(n)$ can be 
 read directly from~\eqref{mat_el_dis}:
\begin{equation}
    \label{jpdf_mat_el}
     p_{\beta}(\av,\bv) = \frac{1}{Z_\beta}\prod_{j=1}^{n-1}\abs{b_j}^{\frac{\beta j}{2} -1} 
     \exp\left[ -\frac12\left(\sum_{j=1}^n \abs{a_j}^2 + \sum_{j=1}^{n-1} 
     \abs{b_j}^2\right)\right],
\end{equation}
where
\begin{equation*}
Z_\beta = \frac{(2\pi)^{2n - 1}}{2^{n- \frac{n(n-1)\beta}{8} -1}}\prod_{k=1}^{n-1} 
    \Gamma\left(\frac{\beta k}{4}\right).
\end{equation*}

We are  now in a position to prove Theorem~\ref{main_theorem}.  By Lemma~\ref{lem:jac1}  and Theorem~\ref{thm:jac} we obtain
\begin{equation}
   \label{eqmes}
    \prod_{j=1}^{n-1}\abs{b_j}^{\frac{\beta j}{2} -2}d^2b_j \prod_{j=k}^n d^2a_k 
    = \abs{\Delta(\blambda)}^\beta \prod_{j=1}^{n}\abs{r_j}^{\frac{\beta }{2}-2}d^2r_j \prod_{j=k}^n d^2\lambda_k 
\end{equation}
The Frobenius norm of a matrix $M=(m_{jk})_{j,k=1,\dotsc,n}$ 
is defined by
\[
\norm{M}_{\mathrm{F}} = \sqrt{\sum_{j,k=1}^n\abs{m_{jk}}^2} = \Tr(MM^*).
\]
Recall that the Schur decomposition is given by
\begin{equation}
    \label{schur}
    T = Q\left(\Lambda + N\right)Q^*,
\end{equation}
where $N$ is a strictly upper-triangular matrix, $\Lambda \in \Lambda(n)$ and $Q$ is unitary.  If $Y$ is a unitary matrix, then $\norm{T}_{\mathrm{F}}$ is invariant under
the similarity transformation $T \mapsto Y^*TY$; therefore, the exponent 
in~\eqref{jpdf_mat_el} can be written as Therefore, the exponent 
in~\eqref{jpdf_mat_el} can be written as
\begin{equation*}
    \begin{split}
         \norm{T}_{\mathrm{F}}^2 & = \sum_{j=1}^n \abs{a_j}^2 + \sum_{j=1}^{n-1}\abs{b_j}^2 + n - 1\\
                                & = \sum_{j=1}^n \abs{\lambda_j}^2 + \sum_{1 \le j < k \le n} \abs{N_{jk}}^2
\end{split}
\end{equation*}
where $T \in \T'(n)$.   We now define the positive function
\begin{equation}
    \label{exponent}
    \begin{split}
        g(\blambda,\rv) &= \frac{1}{2}\left(\sum_{j=1}^n \abs{a_j}^2 + \sum_{j=1}^{n-1}\abs{b_j}^2 
        + n - 1\right) - \frac{1}{2}\sum_{j =1}^n \abs{\lambda_j}^2 \\
        &= 
        \frac12 \sum_{1 \le j < k \le n} \abs{N_{jk}}^2.
    \end{split}    
\end{equation}
By Theorem~\ref{theorem:inversion}, the matrix elements $a_j$ and $b_j$ are rational functions of $\blambda$ and $\mathbf{r}$; therefore, by Eq.~\eqref{exponent}, $g(\blambda,\rv)$ depends only on $\blambda$ and $\rv$.
Combining Eqs.~\eqref{jpdf_mat_el},~\eqref{eqmes} 
and~\eqref{exponent} gives
\begin{gather*}
p_\beta(\av,\bv) \prod_{j=1}^{n}d^2a_j \prod_{k=1}^{n-1}\frac{d^2b_k}{\abs{b_k}} 
=P_\beta(\blambda,\rv) \prod_{j=1}^{n}d^2\lambda_j \prod_{k=1}^{n-1}d^2r_k\nonumber\\
= \frac{1}{Z_\beta}\exp\left( -\frac12 \sum_{j=1}^n \abs{\lambda_j}^2 - g(\blambda,\rv)\right)\abs{\Delta(\blambda)}^\beta 
\prod_{j=1}^{n}\abs{r_j}^{\frac{\beta }{2}-2}  \prod_{j=1}^{n}d^2\lambda_j \prod_{k=1}^{n-1}d^2r_k
\end{gather*}
Integrating over $\rv$ leads to Eqs~\eqref{eq:EigDens1} and~\eqref{factor}.

In order to show that 
\begin{equation}
\label{order}
g(\blambda, \rv) = O(\abs{\lambda_k}^4),  \quad \lambda_k \to \infty, \quad k=1,\dotsc,n,
\end{equation}
we observe that from Eq.~\eqref{eq:sol_rec_r1} and~\eqref{eq:sol_rec_rn} 
we have
\begin{subequations}
\begin{align}
\label{abrl_1}
     \rv_{j+1} &= O\left(\lambda_k^{j}\right), & \lv_{j+1} &= 
     O\left(\lambda_k^{- j }\right),\\
\label{abrl_2}
     a_{n-j+1} &= O\left(\lambda_k\right), &   b_{n-j} &= O\left(\lambda^2_k\right), 
\end{align}
\end{subequations}
as $\lambda_k \to \infty$ and $j=1,\dots,n-1$. Substituting \eqref{abrl_2} into 
Eq.~\eqref{exponent} we arrive at Eq.~\eqref{order}.

The equality
\[
g(\blambda e^{i\psi},\rv) = g(\blambda,\rv), \quad \psi \in [0,2\pi)
\]
follows from the observation that from Eqs.~\eqref{eq:anj} and~\eqref{eq:bnj} we have
\begin{equation*}
    \abs{a_{n-j+1}} = \abs{ \rv_j^t\Lambda \lv_j}, \quad 
    \abs{b_{n-j}} = \abs{\rv_{j+1}^t \Lambda \lv_{j}}, \quad j=1,\dotsc,n.
\end{equation*}
Since only the absolute values of the $a_j$'s and $b_j$'s appear in the right-hand side of Eq.~\eqref{jpdf_mat_el},
the map $\blambda \mapsto \blambda e^{i\psi}$ leaves  the \emph{j.p.d.f.} of the eigenvalues~\eqref{eq:EigDens1} invariant.  This completes the proof of Theorem~\ref{main_theorem}. 

It is worth noting that formulae~\eqref{eq:EigDens1} and~\eqref{factor} imply that the eigenvalues and eigenvectors are not statistically independent, since their 
\emph{j.p.d.f.} does not factorise as the product of a function depending only on the
eigenvalues and one depending only on the eigenvectors.

\section*{Acknowledgements}
We are very grateful to Gernot Akemann, Guillaume Dubach, Tamara Grava, Joseph Najnudel and Patricia P\"{a}\ss ler for many valuable discussions.
This research was funded in part by EPSRC grant no. EP/T517872/1.
For the purpose of open access, the author has applied a ‘Creative Commons Attribution' 
(CC BY) public copyright licence to any Author Accepted Manuscript (AAM) version arising from this submission.

\appendix

\section{The multiplicity of the eigenvalues of a tridiagonal matrix}
\label{multiplicity}
It is tempting to think of the matrices~\eqref{eq:tridiagmatrix} as the complex
analogue of Jacobi matrices.  Certain properties that are linked to the 
tridiagonal structure 
carry through; however, this analogy must be taken with some care.  
For example, the spectrum of
Jacobi matrices is simple. This is not necessarily the case for matrices in
 $\T(n)$; a counterexample is 
$\left (\begin{smallmatrix} 6 & 1\\ -16 & -2\end{smallmatrix} \right)$.  
However, it is still 
true that each eigenvalue has at most one eigenvector, even if its 
algebraic multiplicity is greater than one.
\begin{theorem}
    The geometric multiplicity of an eigenvalue of $T\in \T(n)$ is at most one.
\end{theorem}
\begin{proof}
    The canonical vector 
    \[
    \mathbf{e}_1 = \begin{pmatrix} 1 \\  \vdots \\ 0 \end{pmatrix}
    \]
    is cyclic for $T$; therefore, 
    \begin{equation}
       \label{lin_ind}
       \mathbf{e}_1, \, T\mathbf{e}_1, \dotsc,\, T^{n-1}\mathbf{e}_1
    \end{equation}
    are linearly independent. As a consequence, $T^n\mathbf{e}_1$ has a 
    unique representation as a linear combination of the vectors~\eqref{lin_ind}.  Equivalently, 
    there exists a unique monic polynomial $\chi_n$ of degree $n$ such that $\chi_n(T)\mathbf{e}_1=0$. Furthermore, $\chi_n(T)$ is the polynomial in $T$ of 
    least degree that annihilates $\mathbf{e}_1$. Next, let
    $\phi$ be the minimal polynomial of $T$. It satisfies the following properties:
    \begin{enumerate}
      \item $\phi(T)=0 \implies \phi(T)\mathbf{e}_1=0$;
      \item by the Calyley-Hamilton theorem, the degree of $\phi$ cannot exceed $n$;
      \item $\phi$ is a divisor of the characteristic polynomial of $T$.
    \end{enumerate}
    Thus, it follows that $\chi_n = \phi$ is both the minimal and 
    the characteristic polynomial of $T$. Thus, the spectrum of $T$ is either simple, or, if  $\lambda$ is a multiple root of $\chi_n$, then it has only one block in the Jordan
    normal form of $T$.
    \end{proof}

\section{The spectral parameters for matrices with off-diagonal elements equal to zero}
\label{zero_off_d}

In this appendix we generalize  the proof of Theorem~\ref{theorem:inversion} to the set $B$ of matrices $T\in \T(n)$
whose off-diagonal elements $b_j$, $j=1,\dotsc,n-1$ are not all different from zero. Notice that when 
$T \to T_0$, where $T_0 \in B$, the Jacobian in Theorem~\ref{thm:jac} tends to zero; therefore, the map is not invertible in $B.$
In order to simplify the analysis, we will  restrict to the case when only one of the off-diagonal
elements of $T\in \T(n)$ is zero. The generalization to matrices with two or more of the $b_j$s 
are zero is analogous, but more involved computationally. 

Assume that $b_{n-k}=0$. Fix $\Lambda \in \Lambda(n)$ and two vectors
\begin{equation}
    \label{r_k_and_r_nk}
            \rv_k = (r_1,\dotsc,r_k) \in  \Hp_k\quad \text{and} \quad
            \rv_{n-k} = (r_1,\dotsc,r_{n-k}) \in \Hp_{n-k}.
\end{equation}
We want to show that there exist an $R \in R(n)$ and a $T \in \T(n)$ that  depend uniquely on $\Lambda,$ 
$\rv_k$ and $\rv_{n-k}$ and  such that $T=R\Lambda R^{-1}$, excepts for $\rv_k$ and $\rv_{n-k}$ belonging 
to an exceptional set of measure zero in $\Hp_k$ and $\Hp_{n-k},$ respectively. 

Recall that the $\Phi_j$s and $\chi_j$s are the characteristic polynomials of the submatrices $T_j$ and $U_j$ defined in Eqs.~\eqref{princ_subm} and~\eqref{princ_subm_2}, respectively.  
One can verify by direct calculation that
\begin{equation}
    \label{RB}
        \begin{pmatrix} 1 \\ \Phi_1\bigl(\lambda_j\bigr) \\ \vdots \\ \Phi_{k-1}\left(\lambda_j\right) \\
        0 \\ \vdots \\ 0 \end{pmatrix},  \quad j=1,\dotsc, k,
\end{equation}
are right-eigenvectors of $T.$ Similarly, we have that
\begin{equation}
    \label{RinvB}
        \Bigl( \overbrace{\begin{matrix} 0 & \cdots & 0 \end{matrix}}^{k} \;\; \begin{matrix}
        \chi_{n-k - 1}\bigl(\lambda_{j}\bigr) & \cdots & \chi_1\bigl(\lambda_j\bigr) &  
        1 \end{matrix} \Bigr), \qquad j = k+1, \dotsc,n
\end{equation}
are left-eigenvectors. Hence, last $n-k$ elements of the first column of $R^{-1}$ are zero.
Therefore, neither Lemma~\ref{parameterization} nor Theorem~\ref{theorem:inversion} 
applies. 

Let $R_{T_k}$ and $R_{U_{n-k}}$ denote the matrices of the left and of the right-eigenvectors of 
$T_k$ and $U_{n-k}$.  Furthermore, denote by $\Lambda_k$ the submatrix of $\Lambda$ obtained by keeping
the first $k$ rows and column of $\Lambda$. Similarly, $\Pi_{n-k}$ is the submatrix of $\Lambda$ obtained
by keeping the last $n-k$ rows of $\Lambda.$ We now from Theorem~\ref{theorem:inversion} that 
$T_k$, $U_{n-k}$, $R_{T_k}$ and $R_{U_{n-k}}$ can be reconstructed uniquely from $\Lambda_k$, $\Pi_{n-k}$,
$\rv_k$ and $\rv_{n-k}.$ It follows form Eqs.~\eqref{RB} and~\eqref{RinvB} that the matrix of the right-eigenvectors 
of $T$ has the form
\begin{equation}
    \label{form_r_eig}
          R=\left( \begin{array}{c|c}
                    R_{T_k} & W_{k \, n-k} \\
                    \hline
                    \mathbf{0}_{n-k \, k} & R_{U_{n-k}}
                \end{array}
        \right),    
\end{equation}
where $\mathbf{0}_{n-k \,k}$ is a $k \times (n-k)$ matrix whose entries are zero and $W_{n-k \, k}$ is
a $k \times (n-k)$ matrix whose elements can be derived the first $k$ equations in the recurrence 
relations~\eqref{cp_sm1}.  The same analysis leads to the conclusion that the matrix of the left-eigenvectors
of $T$ has the same structure as~\eqref{form_r_eig}, that is
\begin{equation*}
    R^{-1} = \left( \begin{array}{c|c}
                    R_{T_k}^{-1} & Y_{k \, n-k} \\
                    \hline
                    \mathbf{0}_{n-k \, k} & R^{-1}_{U_{n-k}}
                \end{array}
        \right),    
\end{equation*}
where $Y_{n-k\,k}$ is derived from the first $n-k$ equations in the recurrence~\eqref{rec_rel2}. 

\section{An alternative proof of Lemma \ref{lem:jac1}}
\label{sec_proof}
The results in this appendix are generalizations to complex tridiagonal
matrices of well known properties of Jacobi matrices (see, \emph{e.g.}~\cite{DE02,Par98}).

Write
\begin{equation*}
\widetilde{T}_{i,j}=\begin{pmatrix}
a_i & c_{i-1} &  &  \\
\widetilde{b}_{i-1} & a_{i-1} & \ddots &  \\
 & \ddots & \ddots & c_j \\
 &  & \widetilde{b}_j & a_j 
\end{pmatrix},
\end{equation*}
where $1\le j \le i \le n $.  We call $\widetilde{T}_{ij}$ a 
\emph{principal submatrix} of the matrix $\widetilde{T}$ defined in 
Eq.~\eqref{eq:gen_trid}.  Furthermore, set 
\[
\chi_{i,j}(\mu)=\det(\mu I-\widetilde{T}_{i,j})
\]
with the convection $\chi_{0,1}(\mu) = \chi_{n,n+1}(\mu) =1$ and $\chi=\chi_{n,1}$.
Let $\widetilde{T}=R\Lambda R^{-1}$ be the spectral decomposition of $\widetilde{T}$,
where $\Lambda=\diag(\lambda_1,\dots,\lambda_n)$. The columns of $R$ are the right-eigenvectors of $\Tw$, which 
we denote by $\mathbf{R}_j$, $j=1,\dotsc,n$.  Similarly, we write $\mathbf{L}_j^t$ for the 
left-eigenvectors of $\Tw$, which are the rows $R^{-1}$.  Recall that if $M \in \GL(n,\mathbb{C})$ is
invertible, then the \emph{adjugate} of $M$ is defined by
\[
\adj(M) = \det(M)M^{-1}  \quad \text{and} \quad \adj(M)_{ij} = \frac{\partial \det(M)}{\partial M_{ji}}.
\]

\begin{lemma}
\label{lem:adjoint}
We have
\[
\adj\left(\lambda_jI-\widetilde{T}\right)
=\chi'(\lambda_j)\mathbf{R}_j\mathbf{L}_j^t, \quad j=1,\dotsc,n.
\] 
\end{lemma}
\begin{proof}
Let $\mu\neq\lambda_j$, for all $j=1,\dotsc,n$, so that the 
resolvent $(\mu I-\widetilde{T})^{-1}$ is well defined. We have
\begin{equation}
\label{eq:appendix:limit}
\begin{split}
    \adj \left(\mu I-\widetilde{T}\right)&=
    \det(\mu I-\widetilde{T})(\mu I-\widetilde{T})^{-1}\\
    &=\chi(\mu)(\mu I-R\Lambda R^{-1})^{-1} \\
    &=R\Delta(\mu)R^{-1},
\end{split}
\end{equation}
where $\Delta(\mu)=\diag(\delta_1(\mu),\dotsc,\delta_n(\mu))$ and
\begin{equation*}
\delta_k(\mu)=\frac{\chi(\mu)}{\mu-\lambda_k}=
\prod_{\substack{i=1\\i\neq k}}^n(\mu-\lambda_i).
\end{equation*}
Taking the limit as $\mu\rightarrow\lambda_j$ of both sides of 
$\eqref{eq:appendix:limit}$ yields
\[
\text{adj}\left(\lambda_j I-\widetilde{T}\right)
=R\Delta(\lambda_j)R^{-1},
\]
where
\begin{align*}
\delta_k(\lambda_j)=\begin{cases}
0&\text{for $k\neq j$},\\
\prod_{i\neq j}(\lambda_j-\lambda_i)=\chi'(\lambda_j)&\text{for $k=j$.}
\end{cases}
\end{align*}
\end{proof}

\begin{proposition}
\label{theorem:spec}
Let $R_{rj}=\mathbf{e}_r^t\mathbf{R}_j$ and $L_{js}=\mathbf{L}_j^t\mathbf{e}_s$,
where $j,r,s=1,\dotsc,n$. Then, we have
\begin{equation}
\label{devrel}
    \chi'(\lambda_j)R_{rj}L_{j s} =\begin{cases}
    \chi_{n-s,1}(\lambda_j)c_{n-r}\chi_{n,n-r+2}(\lambda_j)&
    \text{if $s-r=1$, }\\
    \chi_{n-s,1}(\lambda_j)c_{n-s+1}\dotsm c_{n-r}\chi_{n,n-r+2}(\lambda_j)&
    \text{if $s-r>1$, }\\
    \chi_{n-r,1}(\lambda_j) \widetilde{b}_{n-s}\chi_{n,n-s+2}(\lambda_j)&\text{if $r-s=1$,}\\
    \chi_{n-r,1}(\lambda_j)\widetilde{b}_{n-r+1}\dotsm \widetilde{b}_{n-s}\chi_{n,n-s+2}(\lambda_j)&\text{if $r-s>1$,}\\
    \chi_{n-r,1}(\lambda_j)\chi_{n,n-r+2}(\lambda_j) & \text{if $r=s$.}
    \end{cases}
\end{equation}
\end{proposition}
\begin{proof}
The right hand side of Eq.~\eqref{devrel} is obtained by computing the matrix elements of $\adj\left(\lambda_jI-\widetilde{T}\right)$ using the Laplace expansion.
The left hand side is given by Lemma~\ref{lem:adjoint}.
\end{proof}
Some particular cases of Proposition~\ref{theorem:spec} are
\begin{equation}
\label{rel_eq}
    \begin{split}
       \chi'(\lambda_j)R_{1 j}L_{jn} &= c_1 \dotsm c_{n-1},\\
       \chi'(\lambda_j)R_{n j}L_{j 1} & = \widetilde{b}_1\dotsm \widetilde{b}_{n-1},\\
       \chi'(\lambda_j)R_{1 j}L_{j1 } & = \chi_{n-1, 1}(\lambda_j),\\
       \chi'(\lambda_j)R_{n j}L_{j n }& = \chi_{n,2}(\lambda_j). 
    \end{split}
\end{equation}
         
We are now in a position to prove Lemma \ref{lem:jac1}. Let $T\in \T(n)$ be
the matrix introduced in Eq.\eqref{eq:tridiagmatrix} and let $T=R\Lambda R^{-1}$ be
its spectral decomposition. Consider the sequence of monic 
polynomials constructed from the recurrence relation
\begin{equation}
\label{rec_rel}
\chi_k(x) = (x - a_k)\chi_{k-1}(x) -b_{k-1}\chi_{k-2}(x), \quad k=1,\dotsc,n,
\end{equation}
with initial conditions $\chi_{0}=1$, $\chi_{-1}=0$ and $b_0$ undetermined. 
Now, $\lambda_j^{(k)}$ is a zero of 
$\chi_k(x)$ if 
\[
  \left(  \chi_{k-1}\left(\lambda^{(k)}_j\right), \chi_{k-2}\left(\lambda^{(k)}_j\right),
    \dotsc, \chi_0\left( \lambda_j^{(k)}\right) \right)
\]
is a left eigenvector of 
\begin{equation*}
T_{k\,1} =\begin{pmatrix}
a_k & 1 &  &  \\
b_{k-1} & a_{k-1} & \ddots &  \\
 & \ddots & \ddots & 1 \\
 &  & b_1 & a_1 
\end{pmatrix},
\end{equation*}
which is the principal submatrix of $T$ obtained by keeping the last $k$ rows
and columns. It follows that\footnote{The characteristic 
polynomials $\chi_k$ are defined differently from those introduced in Eq.~\eqref{cp_sm}. In Sec.~\ref{sec:Jacobian} they are defined as the
characteristic polynomials of the principal submatrix obtained by keeping the 
\emph{first} $k$ rows and columns. 
Definition~\eqref{cp_sma} is more convenient for the calculations in 
this appendix.} 
\be
\label{cp_sma}
\chi_k(x) = \chi_{k,1}(x)= \det\left(x I - T_{k,1}\right) = 
\prod_{j=1}^k\left(x-\lambda_j^{(k)}\right).
\ee
Write
\begin{align}
\label{eq:twoterm}
\prod_{\substack{1\leq i \leq k\\ 1\leq j\leq k-1}}
\left(\lambda_i^{(k)}-\lambda_j^{(k-1)}\right)=\prod_{i=1}^k\chi_{k-1}\left(\lambda_i^{(k)}\right)=\prod_{i=1}^{k-1}\chi_{k}\left(\lambda_i^{(k-1)}\right).
\end{align}
From the recurrence relation~\eqref{rec_rel}, we have 
\[
\prod_{i=1}^{k-1}\chi_{k}\left(\lambda_i^{(k-1)}\right)=(-b_{k-1})^{k-1}
\prod_{i=1}^{k-1}\chi_{k-2}\left(\lambda_i^{(k-1)}\right).
\]
Then, the relation~\eqref{eq:twoterm} leads to  
\[
\prod_{i=1}^{k-1}\chi_{k}\left(\lambda_i^{(k-1)}\right)=(-b_{k-1})^{k-1}
\prod_{i=1}^{k-2}\chi_{k-1}\left(\lambda_i^{(k-2)}\right)
\]
Therefore, by iterating the last equality we arrive at 
\begin{equation}
   \label{prod_bj}
   \begin{split}
  \prod_{i=1}^n\chi_{n-1}\left(\lambda_i^{(n)}\right) &= \prod_{i=1}^{n-1}\chi_{n}\left(\lambda_i^{(n-1)}\right)\\
    & =(-b_{n-1})^{n-1} \prod_{i=1}^{n-2}
    \chi_{n-1}\left(\lambda_i^{(n-2)}\right)\\
    &=(-b_{n-1})^{n-1}(-b_{n-2})^{n-2}\prod_{i=1}^{n-3}
    \chi_{n-2}\left(\lambda_i^{(n-3)}\right)\\
    &=\dotsb = (-1)^{\frac{n(n-1)}{2}}\prod_{j=1}^{n-1}b_j^j.
\end{split}
\end{equation}

Consider the derivative
\[
\chi'_{n}(x)=\sum_{i=1}^n
\prod_{\substack{j=1\\j\neq i}}^n(x-\lambda_j)
\]
and take the product 
\be
\label{prod_der}
\prod_{i=1}^n\chi_{n}'(\lambda_i)=(-1)^\frac{n(n-1)}{2}\Delta(\blambda)^2.
\ee
If the matrix $R$ satisfies Lemma~\ref{parameterization}, then   
$R_{1j}=r_j$ and $L_{j1}=1$ for all $j=1,\dotsc,n$.  Therefore, 
the third relation in Eq.~\eqref{rel_eq} gives
\[
r_{j}=\frac{\chi_{n-1}(\lambda_j)}{\chi_{n}'(\lambda_j)},
\]
which implies
\be
\label{prod_rj}
\prod_{j=1}^nr_{j}=\prod_{j=1}^n
\frac{\chi_{n-1}(\lambda_j)}{\chi_{n}'(\lambda_j)}.
\ee
Since $\lambda_j^{(n)}=\lambda_j$, combining Eqs.~\eqref{prod_bj}, \eqref{prod_der} and \eqref{prod_rj} gives
\[
\Delta(\blambda)^2=\frac{\prod_{k=1}^{n-1}b_k^k}{\prod_{k=1}^nr_{k}}.
\]

\bibliographystyle{amsplain}
\bibliography{complexbeta}

\providecommand{\bysame}{\leavevmode\hbox to3em{\hrulefill}\thinspace}
\providecommand{\MR}{\relax\ifhmode\unskip\space\fi MR }
\providecommand{\MRhref}[2]{%
  \href{http://www.ams.org/mathscinet-getitem?mr=#1}{#2}
}
\providecommand{\href}[2]{#2}
\begin{thebibliography}{10}

\bibitem{AMP22}
G.~Akemann, A.~Mielke, and P.~P\"{a}{\ss}ler, \emph{Spacing distribution in the
  two-dimensional {C}oulomb gas: {S}urmise and symmetry classes of
  non-hermitian random matrices at noninteger $\beta$}, Phys. Rev. E
  \textbf{106} (2022), 014146.

\bibitem{AZ97}
A.~Altland and M.~R. Zirnbauer, \emph{Nonstandard symmetry classes in
  mesoscopic normal-superconducting hybrid structures}, Phys. Rev. B
  \textbf{55} (1997), 1142--1161.

\bibitem{AS21}
S.~Armstrong and S.~Serfaty, \emph{Local laws and rigidity for {C}oulomb gases
  at any temperature}, Ann. Prob. \textbf{49} (2021), 46--121.

\bibitem{BEY12}
P.~Bourgade, L.~Erd\H{o}s, and H.~T. Yau, \emph{Bulk universality of general
  $\beta$-ensembles with non-convex potential}, J. Math. Phys. \textbf{53}
  (2012), 095221.

\bibitem{BEY14}
\bysame, \emph{Universality of general $\beta$-ensembles}, Duke Math. J.
  \textbf{163} (2014), 1127--1190.

\bibitem{BFS07}
J.~Breuer, P.~J. Forrester, and U.~Smilansky, \emph{Random discrete
  {S}chrödinger operators from random matrix theory}, J. Phys. A: Math. Theo.
  \textbf{40} (2007), F161--F168.

\bibitem{BF25}
S.-S. Byun and P.~J. Forrester, \emph{Progress on the study of the {G}inibre
  ensembles}, KIAS Springer Series in Mathematics, vol.~3, Springer Nature,
  Singapore, 2025.

\bibitem{Due01}
E.~Due{\~{n}}ez, \emph{Random matrix ensembles associated to compact symmetric
  spaces}, Ph.D. thesis, Princeton University, 2001.

\bibitem{Due04}
\bysame, \emph{Random matrix ensembles associated to compact symmetric spaces},
  Commun. Math. Phys. \textbf{244} (2004), 29--61.

\bibitem{DE02}
I.~Dumitriu and A.~Edelman, \emph{Matrix models for beta ensembles}, J. Math.
  Phys. \textbf{43} (2002), 5830--5847.

\bibitem{DF10}
I.~Dumitriu and P.~J. Forrester, \emph{Tridiagonal realization of the
  antisymmetric {G}aussian $\beta$-{E}nsemble}, J. Math. Phys. \textbf{51}
  (2010), 093302.

\bibitem{Dys62a}
F.~J. Dyson, \emph{Statistical theory of the energy levels of complex systems.
  {I}}, J. Math. Phys. \textbf{3} (1962), 140--156.

\bibitem{Dys62b}
\bysame, \emph{Statistical theory of the energy levels of complex systems.
  {II}}, J. Math. Phys. \textbf{3} (1962), 157--165.

\bibitem{Dys62c}
\bysame, \emph{Statistical theory of the energy levels of complex systems.
  {III}}, J. Math. Phys. \textbf{3} (1962), 166--175.

\bibitem{Dys62e}
\bysame, \emph{The threefold way. algebraic structure of symmetry groups and
  ensembles in quantum mechanics}, J. Math. Phys. \textbf{3} (1962),
  1199--1215.

\bibitem{ES07}
A.~Edelman and B.~D. Sutton, \emph{From random matrices to stochastic
  operators}, J. Stat. Phys. \textbf{127} (2007), 1121--1165.

\bibitem{For10}
P.~J. Forrester, \emph{Log-gases and random matrices}, Princeton University
  Press, 2010.

\bibitem{For21}
\bysame, \emph{Dyson’s disordered linear chain from a random matrix theory
  viewpoint}, J. Math. Phys. \textbf{62} (2021), 103302.

\bibitem{FR05}
P.~J. Forrester and E.~M. Rains, \emph{Interpretations of some parameter
  dependent generalizations of classical matrix ensembles}, Probab. Theory
  Relat. Fields \textbf{131} (2005), 1--61.

\bibitem{FR06}
\bysame, \emph{Jacobians and rank 1 perturbations relating to unitary
  {H}essenberg matrices}, IMRN \textbf{2006} (2006), 4836.

\bibitem{FS03}
Y.~V. Fyodorov and H.-J. Sommers, \emph{Random matrices close to {H}ermitian or
  unitary: overview of methods and results}, J. Phys. A: Math. Gen. \textbf{36}
  (2003), 3303--3347.

\bibitem{Gin65}
J.~Ginibre, \emph{Statistical ensembles of complex, quaternion, and real
  matrices}, J. Math. Phys. \textbf{6} (1965), 440--449.

\bibitem{GGGM23}
T.~Grava, M.~Gisonni, G.~Gubbiotti, and G.~Mazzuca, \emph{Discrete integrable
  systems and random {L}ax matrices}, J. Stat. Phys \textbf{190} (2023), 1--35.

\bibitem{GM23}
T.~Grava and G.~Mazzuca, \emph{Generalized {G}ibbs ensemble of the
  {A}blowitz–{L}adik lattice, {C}ircular $\beta$-{E}nsemble and double
  confluent {H}eun equation}, Commun. Math. Phys. \textbf{399} (2023),
  1689--1729.

\bibitem{KN04}
R.~Killip and I.~Nenciu, \emph{Matrix models for circular ensembles}, IMRN
  \textbf{2004} (2004), 2665--2701.

\bibitem{KRV16}
M.~Krishnapur, B.~Rider, and B.~Virág, \emph{Universality of the stochastic
  {A}iry operator}, Commun Pure Appl. Math. \textbf{69} (2016), 145--199.

\bibitem{MPS20}
S.~N. Majumdar, A.~Pal, and G.~Schehr, \emph{Extreme value statistics of
  correlated random variables: A pedagogical review}, Phys. Rep. \textbf{840}
  (2020), 1--32.

\bibitem{Meh04}
M.~L. Mehta, \emph{Random matrices}, third ed., Elsevier Inc., 2004.

\bibitem{NS15}
S.~R. Nodari and S.~Serfaty, \emph{Renormalized energy equidistribution and
  local charge balance in {2D} {C}oulomb systems}, IMRN \textbf{2015} (2015),
  3035--3093.

\bibitem{Par98}
B.~N. Parlett, \emph{The symmetric eigenvalue problem}, SIAM, 1998.

\bibitem{SS12}
E.~Sandier and S.~Serfaty, \emph{From the {G}inzburg-{L}andau model to vortex
  lattice problems}, Commun. Math. Phys. \textbf{313} (2012), 635--743.

\bibitem{SS15}
\bysame, \emph{{2D} {C}oulomb gases and the renormalized energy}, Ann. Probab.
  \textbf{43} (2015), 2026--2083.

\bibitem{Spo20}
H.~Spohn, \emph{Generalized {G}ibbs ensembles of the classical {T}oda chain},
  J. Stat. Phys. \textbf{180} (2020), 4--22.

\bibitem{Spo21}
\bysame, \emph{Hydrodynamic equations for the {T}oda lattice},
  \texttt{arXiv:2101.06528v1} (2021).

\bibitem{Tro84}
H.~F. Trotter, \emph{Eigenvalue distributions of large {H}ermitian matrices;
  {W}igner's semi-circle law and a theorem of {K}ac, {M}urdock, and
  {S}zeg\H{o}}, Adv. in Math. \textbf{54} (1984), 67--82.

\bibitem{Zir96}
M.~R. Zirnbauer, \emph{Riemannian symmetric superspaces and their origin in
  random‐matrix theory}, J. Math. Phys. \textbf{37} (1996), 4986--5018.

\end{thebibliography}

\end{document}